\RequirePackage{fix-cm}
\documentclass[smallextended]{svjour3}       
\smartqed  
\usepackage{graphicx}
\journalname{}

\usepackage{epic}
\usepackage{nicefrac}
\usepackage{latexsym}
\usepackage{amsmath}
\usepackage{epsfig}
\usepackage{amssymb}
\usepackage{enumerate}
\usepackage{multirow}

\usepackage[mathscr]{euscript}

\newcommand{\PP}{{\mathbb P}}

\newcommand{\cL}{{\mathcal L}}

\newcommand{\cT}{{\mathcal T}}

\begin{document}

\title{A `stochastic safety radius' for distance-based tree reconstruction
}



\author{Olivier Gascuel        \and
        Mike Steel 
}


\institute{M. Steel  (corresponding author) \at
              Biomathematics Research Centre, University of Canterbury, Christchurch, New Zealand \\
              Tel.: +64-3-364-2987\\
              Fax:  +64-3-364-2587\\
              \email{mike.steel@canterbury.ac.nz}           
           \and
           O. Gascuel \at
              Institut de Biologie Computationnelle,
LIRMM, UMR 5506 CNRS \& Universit{\'e} de Montpellier, France \\
							\email{gascuel@lirmm.fr}  
}

\date{Received: date / Accepted: date}

\maketitle

\begin{abstract}
A variety of algorithms have been proposed for reconstructing trees that show the evolutionary relationships between species by comparing
differences in genetic data across present-day taxa.  If the leaf-to-leaf distances in a tree can be accurately estimated, then it is possible to
reconstruct this tree from these estimated distances, using polynomial-time methods such as the popular `Neighbor-Joining' algorithm.
There is a precise combinatorial condition under which distance-based methods are guaranteed to return a correct tree (in full or in part)  based on the requirement that the input distances
all lie within some `safety radius' of the true distances. Here, we explore a stochastic analogue of this  condition, and mathematically establish  upper
and lower bounds on this `stochastic safety radius' for distance-based tree reconstruction methods. Using simulations, we show how this notion provides a new way to 
compare the performance of distance-based tree reconstruction methods. This  may help explain why Neighbor-Joining performs so well, as its stochastic safety radius appears
close to optimal (while its more classical safety radius is the same as many other less accurate methods).

\end{abstract}


\newpage

\section{Introduction}
\label{intro}

A central task in evolutionary biology is the reconstruction of `phylogenetic' (evolutionary) trees from genetic data sampled from present-day species that describe how these species evolved
from a common ancestor. These trees can be estimated from a variety of different types of data, but a common approach involves data that are based on
some measure of `evolutionary distance' between species.  A variety of  fast  (polynomial-time) methods have been devised for building a phylogenetic $X$-tree from an arbitrary distance function $d$ on $X$. The most popular by far is  `Neighbor-Joining,'  (NJ) and the paper \cite{sai} that described this heuristic algorithm has now been cited more than 36,000 times.  

A desirable property of such methods is that when a distance function fits exactly on a binary (fully resolved) tree with branch lengths, then the method will return that underlying tree (up to the placement of the root) and its edge lengths.   Moreover, when a distance function $\delta$ is close to an exact fit on some binary tree $T$, many methods also come with a guarantee
that they will return $T$ when applied to $\delta$.  

How close $\delta$ needs to be to a `tree metric' $d$ depends crucially on $w_{\rm min}$,  the smallest interior edge length of $T$; a distance-based
tree reconstruction method is said to have  `safety radius' $\rho$ if the method is guaranteed to return  the underlying  binary tree $T$ when $\delta$ differs from $d$ by less than $\rho \cdot w_{\rm min}$ on each pair of leaves (a precise definition is given shortly).  This notion was introduced by Kevin Attenson 25 years ago in this journal \cite{att}, where he established that for NJ, this safety radius is $\rho= \frac{1}{2}$; moreover, this is the largest possible safety radius for any distance-based tree reconstruction method.  

While this classical safety radius has provided a precise formal way to compare different tree reconstruction methods, it is not always a good predictor of which method will perform more accurately 
on simulated (or real) data; for instance, methods that perform well (e.g. NJ) often have a safety radius that is equivalent to those of methods that perform  poorly (e.g. the Buneman tree \cite{berry}).   In this paper, we consider a more relaxed, statistically-based version of the purely combinatorial safety radius that treats the differences between observed and expected distances as independent random variables. We develop and apply this notion of a `stochastic safety radius', derive formal upper and lower bounds, and compare different tree reconstruction methods using it.

\subsection{Classical safety radius}
\label{sec:1}
For a  phylogenetic $X$--tree $T$ with positive edge lengths $w$, let $d_{(T,w)}$ denote the associated tree metric on $X$ and let $w_{\rm min}$ denote the minimal interior edge length of $T$. 
A method $M$ for 
reconstructing a phylogenetic $X$--tree from each dissimilarity map $\delta$ on $X$ is said to have a {\em $l_\infty$ safety radius} $\rho_n$ if, for any binary phylogenetic tree $T$ with $n$ leaves we have:
$$\|\delta - d_{(T, w)}\|_\infty < \rho_n \cdot  w_{\rm min} \Longrightarrow M(\delta) = T.$$
Here $\|*\|_\infty$ refers to the largest difference between $\delta$ and $d$ over all pairs from $X$. 
  Beginning with the pioneering work of Atteson \cite{att} it is now well known that no method can have $l_\infty$ safety radius $\rho_n>\frac{1}{2}$ and that certain methods such as NJ and `Balanced Minimimum Evolution' (BME) \cite{pau} achieve this bound for all $n$ \cite{par2}.   However, for other methods,  the $l_\infty$ safety radius is  less than $\frac{1}{2}$ and it can even converge to $0$ as $n$ grows  \cite{gas}, \cite{par2}.
  We will refer to the following constraint:  $$\|\delta - d_{(T, w)}\|_\infty < \frac{1}{2} \cdot  w_{\rm min}$$ as the {\em Atteson $l_\infty$ bound}.

Despite the mathematical elegance of these results, there are two problems associated with the $l_\infty$ safety radius approach. Firstly,  it is a strict combinatorial condition and the 
$l_\infty$ metric is extremely sensitive, particularly for large values of $n$,  since it takes only one pair of taxa to have a  $\delta$ value that conflicts substantially with its $d$ value to result in 
a violation of the safety radius. A second related point is that simulations show that  methods such as NJ often return the correct tree even when the safety radius is 
violated.  One combinatorial approach that goes some way towards addressing this second point was taken in \cite{mih} and developed further in \cite{eic}. This latter paper
also showed that NJ has a $l_2$ safety radius $\frac{1}{\sqrt{3}} \approx 0.5773$ for trees with $n=5$ leaves.

A related ``edge safety radius" approach was also pioneered by Atteson \cite{att}, who showed that ADDTREE \cite{sat} is optimal, whereas NJ is not.
However, simulations show that NJ performs well (i.e. as well as ADDTREE) regarding the edge safety radius, which somewhat contradicts the theory. 
Recent results by Bordewich and Mihaescu \cite{bor} also indicate that this theory has some shortcomings and produces a ranking among methods (Greedy BME/NJ) which is not observed in practice.
(there are also examples where the standard safety radius produces strange ranking, e.g.  UPGMA/LS methods \cite{gas}).

\section{Stochastic safety radius}
\label{sec:2}

Let us regard $\delta$ as
differing from $d$ by a random `error'.  More precisely, we suppose that:
\begin{equation}
\label{deltaeq}
\delta(x,y) = d_{(T,w)}(x,y) + \epsilon_{xy},
\end{equation}
where the $\epsilon_{xy}$ values are  independent  normal random variables, each with a mean of $0$ and  a variance equal to $\sigma^2$.
We refer to this simple model as the {\em random errors model}.

Note that in the context of this random errors model, maximum likelihood estimation (MLE) of a tree is equivalent to the ordinary least squares (OLS) tree-reconstruction method, since the OLS score that this method seeks to  minimise is proportional to minus the log of the likelihood function. Several heuristic approaches have been designed to search for the OLS tree, starting with the seminal 1967 papers of Cavalli-Sforza and Edwards \cite{cav} and Fitch and Margoliash \cite{fit}.

Throughout this paper, we let $N(\mu, \sigma^2)$ denote a normal random variable with mean $\mu$ and variance $\sigma^2$. Thus $\epsilon_{xy}$ has the distribution $N(0, \sigma^2)$. 
The following inequality and asymptotic equality (as $x \rightarrow \infty$) are helpful in the results that follow (their proof is in the Appendix). 
 For $x>0$:
\begin{equation}
\label{helps}
 \frac{1}{2} e^{-x^2/2} < \PP(N(0,1)>x)  \sim  \frac{e^{-x^2/2}}{x\sqrt{2\pi}},
 \end{equation}
 where $\sim$ denotes asymptotical equivalence as $x$ grows.

\subsection{Example: the four-taxon case}

With four taxa, most methods (if not all) will use the ``four-point rule" \cite{zar} and select the topology $xy|wz$ if: $$\delta(x,y)+\delta(w,z)< \min\{\delta(x,w)+\delta(y,z), \delta(x,z)+\delta(y,w)\}.$$
Now, under the random errors model, the three sums 
$$\delta(x,y)+\delta(w,z), \delta(x,w)+\delta(y,z), \delta(x,z)+\delta(y,w),$$
constitute three independent normal random variables, each with variance of $2\sigma^2$. Moreover, if the tree generating $d$ has topology $xy|wz$ with an interior edge of
length $w$, 
then the second and third sum have the same mean, which is larger than the mean of the first sum by $2w$.
In particular $\PP(\delta(x,w)+\delta(y,z) > \delta(x,y)+\delta(w,z))$ and $\PP(\delta(x,z)+\delta(y,w) > \delta(x,y)+\delta(w,z))$ are both equal to
$$\PP(N(2w, 4\sigma^2) >0).$$
Consequently, the probability that the correct tree topology $xy|wz$ is selected from $\delta$ is at least:
\begin{equation}
\label{quar}
1- 2\PP(N(2w, 4\sigma^2) <0) = 1-2\PP(N(0,1) < -w/\sigma).
\end{equation}
Thus, there will be (say) a $\sim98\%$ probability of correctly inferring the tree topology if $\sigma$ is $\frac{1}{2}w$ (or less).

It is interesting to compare this to the $l_\infty$ bound of Atteson \cite{att} for methods such as NJ. 
 Recall that this holds when  $|d(x,y)-\delta(x,y)| <\frac{w}{2}$ for all
six pairs $x,y$. Now, under the random effects model, the probability of these events all holding is exactly:
$$(-w/2 < \PP(N(0, \sigma^2) < w/2))^6 = (1-2\PP(N(0,1) >w/2\sigma))^6.$$
Consequently, for $\sigma$ set equal to $\frac{1}{2}w$, as above, the probability the $l_\infty$ bound holds is $(1-2\PP(N(0,1)>1))^6 \approx 0.1,$ which is 
much lower than  the 98\% probability described previously; moreover, in order to ensure the $l_\infty$ bound holds with 98\% probability, we would need
to reduce $\sigma$ to around $w/6$.

\subsection{Larger values of $n$}

To extend the above analysis from $n=4$ to larger values of $n$, it is useful to allow $\sigma^2$ to depend on $n$ (the reason for this becomes clear shortly).   Specifically, let us write:
\begin{equation}
\label{eq1}
\sigma^2 = \frac{c^2}{\log(n)},
\end{equation}
for some value $c \neq 0$.

Notice that $\sigma$ is converging to zero but very slowly (i.e. larger trees require more accuracy, but not a lot more).  First we consider what happens with the $l_\infty$ bound
as $n$ grows.

\begin{proposition}
\label{attlem}
Under the random errors model, the probability that the Atteson $l_\infty$ bound holds for all pairs $x,y$ in a tree with $n$ leaves converges to 0 for $c> \frac{1}{4}\cdot w_{\rm min}$ and converges to 1 for
$c< \frac{1}{4}\cdot w_{\rm min}$. 
\end{proposition}

\begin{proof}
Let $w^*$ be the minimal interior edge length in $T$.  Then: 
$$\PP(|\delta(x,y) - d_{(T,w)}(x,y)| < w^*/2) = \PP(N(0, \sigma^2) \in (-w^*/2, w^*/2)),$$
where $N(0, \sigma^2)$ refers to a normal random variable with mean 0, and variance $\sigma^2$ given by Eqn. (\ref{eq1}).
Thus:
$$\PP(|\delta(x,y) - d_{(T,w)}(x,y)| < w^*/2) = 1- 2\PP\left(N(0,1) > \frac{w^*}{2c}\sqrt{\log(n)}\right).$$
Now, the $l_\infty$ bound is satisfied precisely when $|\delta(x,y) - d_{(T,w)}(x,y)| < w^*/2$ for all $x,y$, and so the probability of this $l_\infty$ bound occurring -- call it $P_\infty$ -- is  given by:
\begin{equation}
\label{pin}
P_\infty= \left(1- 2\PP\left(N(0,1) > \frac{w^*}{2c}\sqrt{\log(n)}\right)\right)^{\binom{n}{2}}.
\end{equation}
Proposition~\ref{attlem} now follows from the asymptotic equivalence in (\ref{helps}) and the observation that, for any sequence $x_n$ with $x_n \sim \frac{1}{n^\beta \sqrt{\log(n)}}$, we have 
$\left(1-x_n \right)^{\binom{n}{2}}$ which converges to $0$ and $1$ when  $\beta<2$ and  $\beta>2$, respectively, as $n$ grows. 
\hfill$\Box$
\end{proof} 

We now define a `stochastic safety radius' that is scaled in such a way as to allow comparisons that are meaningful even as $n$ tends to infinity.

\bigskip

\noindent {\bf Definition [Stochastic safety radius] }  For any $\eta>0$,  we will say that a distance-based tree reconstruction method $M$ has  {\em $\eta$-stochastic  safety radius} $s=s_n$ if for
every binary phylogenetic $X$-tree $T$ on $n$ leaves, with minimum interior edge length $w_{\rm min}$,  and with the distance $\delta$ on $X$ described by the random errors model, we have: $$c < s \cdot w_{\rm min} \Longrightarrow \PP(M(\delta) = T) \geq 1-\eta.$$

\begin{proposition}
\label{bac}
For any method $M$ that has $l_\infty$ safety radius $\rho_n>0$,  and for any $\eta>0$, there is a value $s=s_n>0$ so that $M$ has $\eta$-stochastic radius (at least) $s$ for all binary trees on $n$ leaves.  Moreover, as $n \rightarrow \infty$ we can take $s_n$ arbitrarily close to $\frac{1}{2} \rho_n$. 
\end{proposition}
\begin{proof}

If $c \geq s \cdot w_{\rm min}$ then, from the analogue of (\ref{pin}) (with $\frac{1}{2}w^*$ replaced by $\rho \cdot w^*$), we have:
$$P_\infty \geq \left(1- 2\PP\left(N(0,1) > \frac{\rho}{s}\sqrt{\log(n)}\right)\right)^{\binom{n}{2}}.$$
Applying the inequality in (\ref{helps}) gives:
$$P_\infty \geq \left (1- n^{-\rho^2/2s^2}\right)^{\binom{n}{2}} \geq 1- \binom{n}{2}n^{-\rho^2/2s^2},$$
and the very last term on the right can be made $< \eta$ by selecting $s= s_n$ sufficiently small. Moreover, as $n \rightarrow \infty$, 
we can take $s_n$ to approach $\frac{1}{2} \rho_n$ for any $\eta>0$.
\hfill$\Box$
\end{proof}

Notice that this proof just sets a lower bound on the $\eta$-stochastic safety radius.   This shows that any method having non-zero $l_\infty$ safety radius (e.g.  $1/2$ for NJ) also has non-zero $\eta$-stochastic safety radius, which is roughly equal to half of the $l_\infty$ safety radius (i.e. 1/4 with NJ). In other words, our  definition provides non-trivial performance criteria for those methods, with a lower bound that is easily computed from the $l_\infty$ results. However, the bound in Proposition~\ref{bac} is very severe in that
it requires that the $l_\infty$ bound to hold for all pairwise distances.  We will see that much better bounds do exist. Moreover, the following definition avoids having to consider the effect of $\eta$, which plays a minor role in all calculations.

\bigskip

\noindent {\bf Definition [Limiting stochastic safety radius]} We say that a distance-based tree reconstruction method $M$ has a {\em limiting stochastic safety radius} (LSSR) $r$ if for every $s<r$ and every $\eta>0$ the $\eta$-stochastic safety radius of $M$ is at least $s$ for all binary trees with sufficiently many leaves.

\bigskip

\section{Theoretical Results}

We first show that the limiting stochastic safety radius of a relatively simple quartet-based approach is considerably larger than the value $\frac{1}{4}$ that is required by Proposition~\ref{bac} to satisfy the 
 Atteson bound (where $\rho = \frac{1}{2}$).  We then present our main theoretical result (Theorem~\ref{mainthm}),   an absolute upper bound on the limiting stochastic safety radius of any distance-based method.

 \begin{proposition}
 \label{pro1}
 There is a distance-based tree reconstruction method which has a limiting stochastic safety radius of $\frac{1}{\sqrt{2}}$.
\end{proposition}
\begin{proof}
We use a result from \cite{pea} (see also \cite{kan}) which provides a tree reconstruction method that can recover a binary tree with $n$ leaves by asking for the 
topology of $\Theta( n \log (n))$ quartets (the questions asked are allowed to depend on the answers obtained up to that point).  Provided that all these quartet trees
are correctly returned, we will infer the correct underlying parent tree. Now, suppose we set $\sigma =  s w^*/\sqrt{\log n}$.  From Eqn. (\ref{quar}),  the probability that any particular quartet is correctly inferred is at least:
$$1-2\PP\left(N(0,1) < -\frac{w\sqrt{\log(n)}}{w^* s}\right),$$
where $w$ is the length of the interior path of the quartet in the parent tree. Since $w/w^* \geq 1$, Boole's inequality implies that the probability  that every particular quartet is correctly inferred is at least:
\begin{equation}
\label{Ceq}
1-2C n \log(n) \PP\left(N(0,1) < -\frac{\sqrt{\log(n)}}{s}\right),
\end{equation}
where $C$ is an upper bound constant in  the $\Theta( n \log (n))$ construction. Notice that this holds even though the quartet decisions are not (stochastically) independent.
The expression in (\ref{Ceq})  now converges to 1 for any value $s<\frac{1}{\sqrt{2}}$ as $n \rightarrow \infty$, by the asymptotic equivalence in (\ref{helps}).
\hfill$\Box$
\end{proof}

\begin{theorem}
\label{mainthm}
\begin{itemize}
\mbox{ } 
No distance-based tree reconstruction method has a limiting stochastic safety radius greater than 1.
\end{itemize}
\end{theorem}
\begin{proof}
The idea of the proof is to show that MLE (maximum likelihood estimation) cannot allow a safety radius with $c>1$ on a subset of trees (with prescribed branch lengths), from which it will follow that no other method could do so either on that subclass of trees (and thereby on all binary trees and with variable edge lengths).

Consider the binary tree $T_{3n}$ on the leaf set $\cL= \bigcup_{i=1}^n \{a_i, a'_i, b_i\}$ of size $3n$, obtained from any fixed binary tree on leaf set  $\{1,2,\ldots, n\}$ by replacing each leaf $i$ 
by the rooted triplet subtree $(a_i,a'_i)b_i$.  For $T_{3n}$  assign length 1 to all the interior edges and to the pendant edges that are incident with leaves of type $a_i$ and $a'_i$ (for all $i$),  and assign the length $2$ to the pendant edges that are incident with leaves of type $b_i$ (for all $i$). 
For a sequence ${\bf t} = (t_1, t_2, \ldots, t_n)$ where $t_i \in \{-1, 1\}$ let $T(\bf{t})$ be the tree obtained from $T_{3n}$ by interchanging the leaf labels $a'_i$ and $b_i$ precisely for each $i$ for which $t_i=-1$  (leaving all edge lengths unchanged -- thus all cherry pendant edges have length 1,while the non-cherry pendant edges have length 2). 
Thus $\cT= \{T_{3n}({\bf t}): {\bf t} \in \{-1,1\}^n\}$ is a set of $2^n$ binary trees, each with a leaf set $\cL$ of $3n$ leaves, and with the prescribed branch lengths described (see Fig.~\ref{fig1}).

\begin{figure}[ht]
\includegraphics[width=0.75\textwidth]{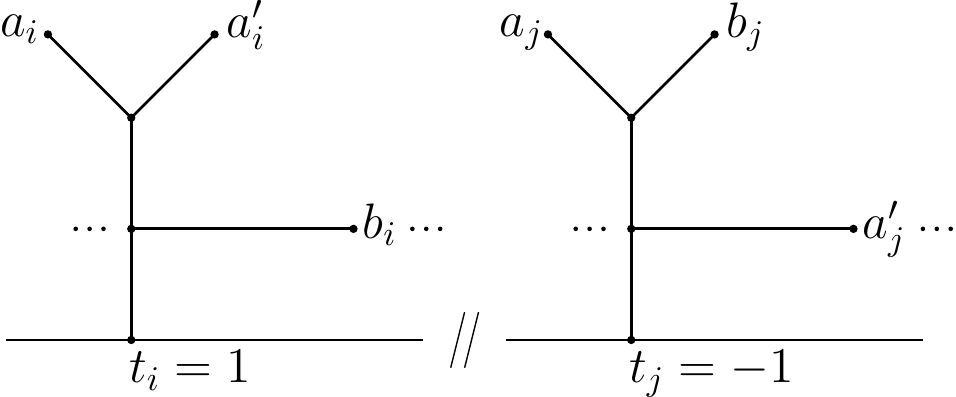}
\caption{The tree $T_{3n}({\bf t}).$}
\label{fig1}
\end{figure}

Let $d^{\bf t}$ denote the tree metric induced on $\cL$ by the tree $T_{3n}({\bf t})$ with its associated branch lengths, and let 
$\delta^{\bf t}$ denote the corresponding distances under the random errors model (so $\delta^{\bf t}= d^{\bf t} + \epsilon$, for a vector $\epsilon$ of independent Gaussians with mean 0  and variance $\sigma^2$ (and independent of ${\bf t}$)).
Notice that we can partition any vector of distances $\delta$ on $\cL$ into two parts $\delta_B$ and $\delta_W$, where $\delta_B$ compares leaves between different triplet-subtrees ($B =$ `between') and $\delta_W$ compares leaves within given triplet-subtrees ($W=$ `within'); formally:
\begin{itemize}
\item $\delta_B$ is the sequence of $\delta$-values
for all pairs $\omega \in B$ where:
$$B = \{ (l_i, l_j):  i,j \in \{1,\ldots, n\},  i \neq j,  l_i \in \{a_i, a'_i, b_i\}, l_j \in \{a_j, a'_j, b_j\}\};$$ 
\item
$\delta_W$ is the sequence of $\delta$-values
for all pairs $\omega\in W$,  where: 
$$W= \{(l_i, l'_i): i \in \{1,\ldots, n\}, l_i, l_i'  \in \{a_i, a'_i, b_i\}, l_i \neq l'_i\}$$ 
\end{itemize}
It is useful to partition the `within' pairs further  as follows.  For each $j \in \{1, \ldots, n\}$, let 
$V(j) = \{(a_j, a'_j), (a_j, b_j)\}$ and let $U(j)= \{(a'_j, b_j)\}$. 

A fundamental observation at this point is that the probability distribution of $\delta^{\bf t}$ on all pairs from $B$ and the pair in  $U(j)$ (for each $j$)  does not depend on ${\bf t}$ at all.
 Moreover,
for pairs from $V(j)$, the dependence of  $\delta^{\bf t}$ on ${\bf t}$ is only via $t_j$. 
By this invariance and the independence assumption in the random errors model, for any $\delta$ and ${\bf t} \in \{-1,1\}^n$, the probability density function for $\delta^{\bf t}$ can be written in the following factored way:
\begin{equation}
\label{likeq}
f(\delta|{\bf t}) = \prod_{\omega \in B} f(\delta_\omega) \cdot \prod_{j=1}^n f(\delta (a'_j, b_j)) \cdot \prod_{j=1}^n f(\delta (a_j, a'_j)|t_j) f(\delta (a_j, b_j)|t_j),
\end{equation}
where, for $\omega \in V(j)$, the terms in the third product are given by:
\begin{equation}
\label{feq0}
f(\delta(\omega)|t_j)  = \frac{1}{\sigma\sqrt{2\pi}} \exp\left (-\frac{(\delta(\omega) - d_{w_j}(\omega))^2}{2\sigma^2}\right),
\end{equation}
in which:
$$d_{w_j}(\omega) = 
\begin{cases}
2, & \mbox{ if } t_j=1 \mbox{ and } \omega = (a_j, a'_j)  \mbox{ or } t_j = -1 \mbox{ and } \omega = (a_j, b_j);\\
4, &  \mbox{ if } t_j=-1 \mbox{ and } \omega = (a_j, a'_j) \mbox{ or } t_j = 1 \mbox{ and } \omega = (a_j, b_j).
\end{cases}
$$

Notice that from our fundamental observation above, the terms in the first two products appearing in (\ref{likeq}) do not depend at all on ${\bf t}$. 

Thus, for any $\delta$, the maximum likelihood estimate of ${\bf t}$ given $\delta$  is the sequence $(t_1, \ldots, t_n)$ where, for each $j \in \{1, \ldots, n\},$
$t_j$ maximises the product $$ f(\delta (a_j, a'_j)|t_j) \cdot f(\delta (a_j, b_j)|t_j).$$

For such a ML estimate ${\bf t}$,  the following inequality must hold for all $j \in  \{1, \ldots, n\}$:
\begin{equation}
\label{feq}
L_j : = \frac{ f(\delta (a_j, a'_j)|t_j)}{f(\delta (a_j, a'_j)|-t_j)} \cdot \frac{f(\delta (a_j, b_j)|t_j)}{f(\delta (a_j, b_j)|-t_j)} \geq 1.
\end{equation}

Assume now that  ${\bf t} = {\bf 1} = (1,1, \ldots, 1)$. Then 
if we let $\Delta_j = \delta(a_j, b_j) - \delta(a_j, a'_j)$, application of (\ref{feq0}) in (\ref{feq}) simplifies  (after some algebra) to the more attractive equation:
\begin{equation}
\label{elegant}
L_j = \exp\left(\frac{2\Delta_j}{\sigma^2}\right).
\end{equation}

To this point, $\delta$ has been an arbitrary distance.  Now, let us further assume that $\delta$ is generated on $T_{3n}({\bf 1})$ under the random errors model; in other words, $\delta = \delta^{\bf t}$ for
${\bf t} ={\bf 1}$. 
We wish to calculate the probability --  call it $p_n$ -- that MLE will correctly estimate the generating tree $T_{3n}({\bf 1})$. By  (\ref{feq}) and independence assumptions in the random errors model, this probability 
$p_n$ satisfies: 
\begin{equation}
\label{feq2}
 p_n \leq \PP(\bigcap_{j=1}^n L_j \geq 1| {\bf t= 1} ) = \prod_{j=1}^n \PP(L_j  \geq 1|{\bf t= 1}).
\end{equation}
Now, from  (\ref{elegant}), $$ \PP(L_j  \geq 1|{\bf t=1}) = \PP(\Delta_j >0|t_j=1),$$
and $\Delta_j$ has a normal distribution with a mean of 2 and a variance of $2\sigma^2$.
Thus: 
\begin{equation}
\label{del}
\PP(\Delta_j >0|t_j=1) = 1 - \PP\left(Z< -\frac{\sqrt{2}}{\sigma}\right),
\end{equation}
where $Z = N(0,1)$ is a standard normal random variable.
Substituting (\ref{eq1}) and (\ref{helps})  into (\ref{del}) gives: 
$$\PP(\Delta_j >0|t_j=1) = 1-\beta_n,$$
where 
$$\beta_n \sim \frac{c}{2\sqrt{\pi \log(n)}} n^{-1/c^2}.$$
Applying this to Eqn. (\ref{feq2}) gives that the probability $p_n$ that MLE correctly estimates the generating tree
satisfies  $$p_n \leq (1- \beta_n)^n.$$ 
Straightforward calculus now shows that as $n\rightarrow \infty$, the sequence $(1- \beta_n)^n$ (and hence $p_n$)  converges to 0 if $c>1$.

This shows that we cannot recover $T_{3n}({\bf 1})$ with an accuracy bounded away from 0 as $n$ becomes large, by using MLE,  if $s>1$ in the definition of the stochastic safety radius (since the interior edges all have length 1, we have $w^*=1$ and so $c>1$ for $s>1$). Moreover, by symmetry, the same conclusion applies to any of the $2^n$ 
trees $T_{3n}({\bf t})$ (there is nothing `special' about ${\bf t} = {\bf 1}$). We now invoke a classic result that  MLE is an estimation method that maximises the average reconstruction accuracy of a discrete parameter when a family of distributions depends on just that parameter ({\em c.f.} Theorem 10.3.1 of
\cite{cas} or Theorem 17.2 of \cite{gui}) --  in our case, the discrete parameter is the vector ${\bf t}$ (which determines the  tree $T_{3n}({\bf t})$).  It follows that for any distance-based reconstruction method, the limiting stochastic safety radius cannot be larger than $1$. 
\hfill$\Box$

\end{proof}

\section{Simulation Results}

We have seen in the previous section that the limiting stochastic safety radius (LSSR) of any algorithm is at most 1 (Theorem \ref{mainthm}), and that a simple quartet algorithm \cite{pea} has a LSSR value at least   $\frac{1}{\sqrt{2}} \approx 0.71$ (Proposition \ref{pro1}). The gap is relatively small between these two bounds and we expect that more sophisticated algorithms have LSSR values that are substantially higher than $\frac{1}{\sqrt{2}}$. In this section, we turn to simulations to study the accuracy of mainstream distance-based methods under the random errors model, with realistic numbers of taxa (previous results are asymptotic). Our goal is to compare these methods and to check how close they come in practice to the 1 bound prescribed by Theorem \ref{mainthm}.  Notice that this theorem was established using the pronged trees of Figure 1. These trees are expected to be difficult for two reasons: (1) all internal branches have the same length and thus no branch is easy; (2) a large number ($2n/3$) of taxon pairs are separated by a single internal branch and are likely to be wrongly exchanged, when trying to infer these trees from noisy data.  Not all tree shapes possess this property; for example, in a perfectly balanced tree, all non-cherry taxon pairs are separated by at least {\em two} internal edges. Thus, we also compare, using simulations, different tree shapes, to establish if some of them are (stochastically) `harder' than the others to reconstruct (depending, perhaps, on the inference method), or whether the opposite is true and all tree shapes seem to be equally difficult.

In the following, we first describe the methods being compared and the comparison criteria, we then study their performance with pronged and other (e.g. perfectly balanced) extreme trees, and, lastly, we use randomly generated tree shapes to obtain average accuracy measures under the random errors model.

\subsection{  Methods tested and comparison criteria}
We ran four standard algorithms using FastME implementation (http://www.atgc-montpellier.fr/fastme/):

\begin{enumerate}[(1)]
\item	  GME+OLS (Greedy minimum evolution with ordinary least squares)  \cite{des} is a greedy algorithm that iteratively adds taxa on a growing tree, minimizing at each step the ordinary least squares (OLS) tree length estimate, in accordance with the OLS version of the minimum evolution principle \cite{rzh}. The performance of this algorithm was analysed by \cite{par2}, who showed that its  $l_\infty$ safety radius tends to 0 with increasing $n$. 
\item	 UNJ (Unweighted NJ) \cite{gas97},  which is the unweighted (OLS) version of NJ, with an $l_\infty$ safety radius of $\frac{1}{2}$ \cite{att}.
\item	 GME+BME \cite{des} which uses the same iterative taxon addition scheme as GME+OLS, but optimizes the balanced version of minimum evolution (BME, \cite{pau}) and has an optimal $l_\infty$ safety radius of $\frac{1}{2}$ \cite{par2}.
\item	 NJ \cite{sai} with  $l_\infty$ safety radius of $\frac{1}{2}$  \cite{att}. We showed \cite{gas2} that NJ greedily minimises BME at each agglomeration step (and not the OLS version of minimum evolution, as was originally suggested).
\end{enumerate}

The aim was to see if there is any difference between the algorithmic schemes (taxon addition versus cherry agglomeration), and between the criteria being optimised (BME vs OLS minimum evolution). Because of the OLS-type noise in the random errors model, better accuracy is expected for OLS-based algorithms (UNJ and GME+OLS). On the other hand, the $l_\infty$ safety radius  differs widely among these algorithms and converges to 0 for GME+OLS; thus the second aim was to check whether Atteson's  predictions are observed in practice.

In addition,  we implemented the quartet method of \cite{pea} using the first algorithm described by \cite{kan} with $O(n^2)$  time complexity and  $O(n\log_2(n))$ quartet queries. The usual four-point rule \cite{sat} \cite{zar} was used to answer the quartet queries. Our aim was to compare the accuracy of this simple algorithm, mostly used for theoretical purposes, to that of algorithms being widely used in phylogenetics, and to check how this algorithm behaves regarding our $\frac{1}{\sqrt{2}}$ bound of Proposition~\ref{pro1}.

Two criteria were used to compare algorithm accuracy: (1) the probability $P_c$ of recovering the entire topology of the simulated tree; (2) the normalized bipartition distance (Robinson-Foulds (RF) \cite{rob}) between the inferred and simulated trees, which is equal to 0 when both trees define the same bipartitions, and to 1 when they do not share any bipartition in common.

\subsection{Algorithm accuracy with pronged and other extreme trees}
    
We used pronged trees (Fig.~\ref{fig1}, used in the proof of Theorem~\ref{mainthm}) with a number of taxa $n=12, 30, 90, 270$ and $810$.  We also used: caterpillar trees with $n=90$ and $270$; perfectly balanced trees with $n=96$ and $384$;  and ``balanced+pronged" trees, where each leaf of a perfectly balanced tree is replaced by the same three-taxon tree as in the pronged trees, with  $n=72$ and 288. Caterpillar and perfectly balanced tree shapes are extreme regarding a number of measurements (e.g. diameter, number of cherries, etc.), and the pronged 3-taxon tree is assumed to make tree inference difficult (see above). In all of these trees, all internal branches had an equal length of 1, which (again) makes tree inference difficult. The length of the external branches was 1 for the caterpillar and balanced trees, and as shown in Fig.~\ref{fig1} for the pronged trees.  The pairwise distances were computed and perturbed by an independent and identically distributed normal noise with standard deviation equal to
 $1/\sqrt{\log(n)}$, i.e.  the highest possible noise level regarding Theorem~\ref{mainthm}, beyond which no algorithm can accurately recover every tree correctly as $n$ grows. The goal was to check if, in these especially difficult conditions, the standard algorithms (e.g. NJ) still show some ability to recover the correct tree. In these conditions, the quartet method had very poor results that are not shown (but see below). For each of the tree shapes and $n$ values, 500 data-sets were generated to obtain average error estimates.  
The results are summarised in Table~1. We see the following:


\begin{table}
\begin{center}
\caption{Algorithm accuracy with the limiting noise level ($\sigma = w^*/\sqrt{\log(n)}$) and extreme tree shapes}
\resizebox{11.75cm}{!}{
\begin{tabular}{|c|c|c|c|c|c|c|c|c|c|}
\cline{1-10}
& &	\multicolumn{2}{|c|}{GME+OLS}&	\multicolumn{2}{|c|}{GME+BME} &	\multicolumn{2}{|c|}{UNJ}& \multicolumn{2}{|c|}{NJ}\\
\cline{2-10}Tree shape & \#taxa	&$P_c$	&RF	&$P_c$	&RF	&$P_c$	&RF	&$P_c$	&RF \\
\cline{1-10} \multirow{5}{*}{pronged (Fig. 1)} & 12 & 0.748 & 0.03111 & 0.750 & 0.03088 & 0.898 & 0.01266 & 0.890 & 0.01377 \\
\cline{2-10}& 30 & 0.734 & 0.01103 & 0.758 & 0.00992 & 0.884 & 0.00444 & 0.872 & 0.00488 \\
\cline{2-10}& 90 & 0.768 & 0.00303 & 0.804 & 0.00243 & 0.888 & 0.00128 & 0.888 & 0.00128 \\
\cline{2-10}& 270 & 0.834 & 0.00069 & 0.840 & 0.00063 & 0.924 & 0.00029 & 0.920 & 0.00030 \\
\cline{2-10}& 810 & 0.856 & 0.00019 & 0.850 & 0.00020 & 0.931 & 0.00009 & 0.944 & 0.00007  \\
\cline{1-10} \multirow{2}{*}{caterpillar} & 90 & 0.736 & 0.00360 & 0.426 & 0.00937 & 0.996 & 0.00005 & 0.966 & 0.00039 \\
\cline{2-10}& 270 & 0.792 & 0.00095 & 0.392 & 0.00348 & 0.998 & 0.00001 &  0.992 & 0.00003 \\
\cline{1-10}\multirow{2}{*}{balanced} & 96 & 0.892 & 0.00133 & 0.884 & 0.00133 & 0.980 & 0.00032 & 0.976 & 0.00037 \\
\cline{2-10}& 384 & 0.892 & 0.00031 & 0.918 & 0.00024 & 0.976 & 0.00009 & 0.976 & 0.00009 \\
\cline{1-10}\multirow{2}{*}{balanced + pronged} & 72 & 0.828 & 0.00281 & 0.846 & 0.00243 & 0.918 & 0.00130 & 0.916 & 0.00136 \\
\cline{2-10}& 288 & 0.854 & 0.00055 & 0.876 & 0.00047 & 0.938 & 0.00022 & 0.938 & 0.00022 \\
\cline{1-10}\noalign{\smallskip}
\end{tabular}
} \end{center}
\label{tab:1}       
\end{table}

\begin{itemize}
\item 
{\bf Pronged trees}  are indeed difficult to reconstruct accurately, compared to perfectly balanced trees. However, for all algorithms, the probability of recovery ($P_c$) increases with $n$. According to this result, the four tested algorithms could have LSSR equal to 1. However, the algorithms are not equivalent in their performance.  We see a clear advantage of the agglomerative scheme (NJ and UNJ) over taxon addition (GME+OLS and GME+BME), a finding that has already been observed in other simulation studies (e.g. \cite{des}). On the other hand, there is no significant difference (considering both $P_c$ and RF) between the algorithms that minimise the OLS version of minimum evolution (GME+OLS and UNJ) and their BME counterparts (GME+BME and NJ, respectively). Notably, we do not see any sign of weakness of GME+OLS, as predicted by its limiting $l_\infty$ safety radius of zero \cite{par2}.

\item
{\bf Caterpillar trees} give another view. Again we observe the clear advantage of the agglomeration scheme that obtains nearly perfect results ($P_c \approx 1$), especially with UNJ, which  is substantially better than NJ regarding both $P_c$ and RF criteria. This latter finding is expected with such unbalanced trees, where the matrix reduction step is better achieved by UNJ, which  accounts for the number of taxa in both agglomerated subtrees, while NJ uses equal weights of $\frac{1}{2}$. Based on these results, NJ and UNJ seem again to have LSSR of (close to) 1. GME-OLS has a lower $P_c$ value, but this increases with $n$. In the opposite, GME-BME not only has low $P_c$ ($< 0.5$), but this decreases with $n$. We have no clear explanation for this poor performance but, based on this result, it is unlikely that GME+BME has LSSR equal to 1.

\item 
{\bf Perfectly balanced trees} confirm again the superiority of the agglomeration scheme, compared to taxon addition. NJ and UNJ are nearly the same, as expected with well-balanced trees (see above). However, although these trees are relatively easy for all algorithms ($P_c \geq 0.9$), we do not see any improvement with larger $n$ values. This questions our previous assumption that LSSR could be equal to 1 for the algorithms tested, unless the convergence towards LSSR is slow.

\item
{\bf Balanced+pronged trees}  are more difficult than balanced trees, as expected. However, for all algorithms, the accuracy increases with $n$.  Algorithm comparisons are consistent with the previous ones: the agglomeration scheme performs better than taxon addition; there is little difference between the OLS and BME versions of algorithms, especially regarding NJ versus UNJ, which are nearly equivalent with such well-balanced trees.

\end{itemize}

	To summarise, NJ and UNJ have remarkably high accuracy with these difficult trees and conditions. These results suggest that these two methods might have an LSSR equal to the optimal value of 1.  Note, however, that the performance of NJ and UNJ with perfectly balanced trees lags a little behind, as their accuracy does not seem to improve when $n$ increases; although, this may be because the convergence is slow. The taxon addition scheme is clearly less accurate and the results of GME-BME with caterpillar trees seem to indicate that this algorithm does not have an optimal LSSR of 1. Even though these conclusions are somewhat speculative, as $n$ remained relatively moderate in all of our experiments, these results provide directions for future investigations on the LSSR of mainstream methods.

\subsection{Algorithm accuracy with random trees}

Up to this point, our experiments have used a limited number of extreme tree shapes. In this subsection, we use random trees in the search for other potentially difficult cases and to estimate the average accuracy of tested algorithms under the random errors model. 	

Tree shapes were generated uniformly at random (the so-called `Proportional-to-Distinguishable-Arrangements (PDA)  model') and all branch lengths were set to 1. The number of taxa was 
$n=10, 30, 90, 270$ and $810$.  The pairwise distances were perturbed by an i.i.d. normal noise with standard deviation equal to $1/\sqrt{\log(n)}$, as in previous experiments. Again, these noise level and trees (with equal branch lengths) make tree inference difficult.


\begin{table}
\begin{center}
\caption{Algorithm accuracy with random trees and a  limiting noise level  ($\sigma = w^*/\sqrt{\log(n)}$)}
\resizebox{11.75cm}{!}{
\begin{tabular}{|c|c|c|c|c|c|c|c|c|}
\cline{1-9}
&	\multicolumn{2}{|c|}{GME+OLS}&	\multicolumn{2}{|c|}{GME+BME} &	\multicolumn{2}{|c|}{UNJ}& \multicolumn{2}{|c|}{NJ}\\
\cline{1-9} \#taxa	&$P_c$	&RF	&$P_c$	&RF	&$P_c$	&RF	&$P_c$	&RF \\
\cline{1-9} 10 & 0.710 & 0.04714 & 0.672 & 0.05542 & 0.904 & 0.01457 & 0.854 & 0.02171 \\
\cline{1-9}  30 &
0.778 & 
0.00918 &
0.800 &
0.00822 &
0.948 &
0.00199 &
0.936 &
0.00251 \\
\cline{1-9} 90 &
0.762 &
0.00328 &
0.830 &
0.00209 &
0.964 &
0.00045 &
0.960 &
0.00051 \\
\cline{1-9}  270 &
0.646 &
0.00164 &
0.800 &
0.00080 &
0.966 &
0.00014 &
0.966 & 
0.00014 \\
\cline{1-9}  810 &
0.498 &
0.00089 &
0.798 &
0.00027 &
0.970 &
0.00004 &
0.970 & 
0.00004 \\
\cline{1-9}\noalign{\smallskip}
\end{tabular}
} \end{center}
\label{tab:2}       
\end{table}

The average results over 500 data-sets are displayed in Table 2. NJ and UNJ results are quite consistent with those obtained with extreme tree shapes ({\em cf} Table 1). With moderate values of $n$, UNJ is slightly more accurate than NJ, as expected with the OLS-type noise used in these simulations. With large values of $n$, UNJ and NJ are nearly perfect ($P_c \approx 1$), and both algorithms become strictly equivalent. Again, the results seem to indicate that both UNJ and NJ could have an optimal LSSR of 1. We cannot exclude that particularly difficult trees do exist, but these must be rare, and in average NJ and UNJ appear to be highly accurate at a noise level that is four times larger than the limit for Atteson's approach to apply.

With taxon addition, the picture is different. GME+OLS has low accuracy that drops when $n$ increases, a finding which may be seen as consistent with the predictions from \cite{par2}  using the $l_\infty$ safety radius. GME+BME performs better but its accuracy is not that high and seems to stabilise around 0.8 when $n$ increases. This observation is probably  explained by the fact that some trees (e.g. caterpillars, Table 1) are difficult for this algorithm. In both cases, the results in Table 2 seem to indicate that neither GME+BME nor GME+OLS has an optimal LSSR of 1.

\subsection{Algorithm accuracy with various noise levels}

We also ran simulations where the noise level varied around the limiting value used in previous experiments 
($\sigma = rw^*/\sqrt{\log(n)}, r=1$), in order to study the sharpness of our bounds with realistic numbers of taxa. Our goal was to check whether the accuracy improves substantially for taxon addition when $r$ is less than 1, and whether the agglomeration algorithms (NJ and UNJ) still show some ability to recover the correct tree when $r$ is larger than 1. We also studied the performance of the quartet method with (relatively) low $r$ values. 
We used the same random trees (with all branch lengths equal to $1=w^*$), as in the previous subsection, but used $r=3/5, 3/4$ and $4/3$ instead of $r=1$.


\begin{table}
\begin{center}
\caption{Algorithm accuracy with various noise levels}
\resizebox{11.75cm}{!}{
\begin{tabular}{|c|c|c|c|c|c|c|c|c|c|}
\cline{1-10}
& &	\multicolumn{2}{|c|}{GME+OLS}&	\multicolumn{2}{|c|}{GME+BME} &	\multicolumn{2}{|c|}{UNJ}& \multicolumn{2}{|c|}{NJ}\\
\cline{2-10}$r$ (noise level)  & \#taxa	&$P_c$	&RF	&$P_c$	&RF	&$P_c$	&RF	&$P_c$	&RF \\
\cline{1-10} \multirow{3}{*}{\nicefrac{3}{4}} &
30 &
0.986 &
0.00051 &
0.992 &
0.00029 &
1.0 &
0.0 &
1.0 &
0.0 \\
 \cline{2-10} & 90 &
0.974 &
0.00029 &
0.986 &
0.00016 &
1.0 &
0.0 &
1.0 &
0.0 \\
 \cline{2-10} & 270 &
0.982 &
0.00007 &
0.996 &
0.00001 &
0.998 & 
$<$0.00001 &
0.998 & 
$<$0.00001 \\
\cline{1-10} \multirow{3}{*}{1} &
30 &
0.778 &
0.00918 &
0.800 &
0.00822 &
0.948 &
0.00199 &
0.936 &
0.00251 \\
\cline{2-10} & 90 &
0.762 & 
0.00328 &
0.830 &
0.00209 &
0.964 &
0.00045 &
0.960 &
0.00051 \\
\cline{2-10} & 270 &
0.646 &
0.00164 &
0.800 &
0.00080 &
0.966 &
0.00014 &
0.966 &
0.00014 \\
\cline{1-10} \multirow{3}{*}{4/3} &
30 &
0.316 &
0.04170 &
0.314 &
0.04177 &
0.670 &
0.01414 &
0.606 &
0.01696 \\
\cline{2-10} &
90 &
0.204 &
0.01972 &
0.240 &
0.01696 &
0.634 &
0.00556 &
0.616 &
0.00590 \\
\cline{2-10} &
270 &
0.054 &
0.01099 &
0.122 &
0.00776 &
0.592 &
0.00202 &
0.560 & 
0.00217\\
\cline{1-10}\noalign{\smallskip}
\end{tabular}
} \end{center}
\label{tab:3}       
\end{table}

The average results over 500 data sets are displayed in Table 3 for the standard algorithms. We see that our previous conclusions are confirmed when $r$ differs from 1: the agglomeration scheme performs better than taxon addition; there is little difference between both minimum evolution versions; however, UNJ is slightly better than NJ (e.g. see $r=4/3$, with both $P_c$ and RF). We also see that the 1 bound prescribed by Theorem 1 is rather sharp with moderate number of taxa: with $r=3/4$, both taxon addition algorithms have high accuracy for all $n$ values, and NJ and UNJ are nearly perfect. Conversely, with $r=4/3$,  the accuracy of all algorithms drops dramatically, especially that of the taxon addition scheme where $P_c$ approaches 0 with $n=270$. 
This confirms  that the limiting optimal bound of Theorem~\ref{mainthm}, obtained with special (``pronged") trees, is robust and found again, at least qualitatively, with random trees and realistic $n$ values.


\begin{table}
\begin{center}
\caption{Accuracy of the quartet method with various noise levels}
\resizebox{8cm}{!}{
\begin{tabular}{|c|c|c|c|c|}
\cline{1-5}
$r$ (noise level) & \multicolumn{2}{|c|}{\nicefrac{3}{5}}&	\multicolumn{2}{|c|}{\nicefrac{3}{4}} \\
\cline{1-5} \#taxa	&$P_c$	&RF	&$P_c$	&RF  \\
\cline{1-5} 30 &
0.902 &
0.556 &
0.556 &
0.09170 \\
\cline{1-5} 90 &
0.946 &
0.479 &
0.479 &
0.10795 \\
\cline{1-5} 270 &
0.952 &
0.370 &
0.370 &
0.12478 \\
\cline{1-5}\noalign{\smallskip}
\end{tabular}
} \end{center}
\label{tab:4}       
\end{table}

Table 4 displays the results of the quartet method for $r=3/5$ and $3/4$ (the results with $r=1$ are quite poor and are not shown). Again, the stochastic radius framework and our bounds of Proposition~\ref{pro1} 
(LSSR $ \geq 1/\sqrt{2} \approx 0.71$) and Theorem~\ref{mainthm} (LSSR $\leq1$) have a good predictive accuracy.  With $r=3/5$, the accuracy is high and increases with $n$, while we observe the opposite with $r=3/4$, which seems to indicate that the quartet method has a LSSR close to our $1/\sqrt{2}$ bound.  Note, however, that when $r$ approaches $\frac{1}{\sqrt{2}}$   from below (e.g. $r = 2/3$ or $r = 0.70$), the accuracy is not that high and does not increase with the moderate values of $n$ used in these simulations. This is most likely explained by the very slow convergence of the bound in Equation (\ref{Ceq}), combined with asymptotic equivalence (\ref{helps}), when $r$ is close to $\frac{1}{\sqrt{2}}$.

\section{Discussion}

	Our simulation results show that the stochastic radius framework introduced in this paper has a good predictive capacity and seems to be robust. The optimal bound ($r = 1$) of Theorem~\ref{mainthm}, which was obtained with special ``pronged" trees, seems to apply to a large variety of trees (no tree is easy or else these are quite rare). Moreover, when the noise level decreases below $r = 1$, the accuracy rises for all algorithms and all values of $n$. The behaviour of the quartet method (with moderate $n$ values) is also consistent with the (limiting) prediction of Proposition 3. We thus believe that the LSSR approach will show a high capacity in predicting algorithm performance in realistic conditions, a property which does not hold in several cases with the $l_\infty$ safety radius, as noted in the Introduction.

An important outcome of the simulations is that the ability to recover single branches may still be high, even when the probability of recovering the entire tree drops due to high noise level; for example, with $r = 4/3$ and $n = 270$, the probability that any given branch is correct is higher than $\sim 0.99$ for all standard algorithms (Table 3, RF values). This strongly suggests the development of  stochastic `edge radius' approaches (analogous to the classical non-stochastic concept considered also by Atteson \cite{att}) which would account for the length of the branch being considered, and thus will not use the worst-case approach used here in several places, where all branches have the same length.  In other words,  the tree $T$ may have some very short edges, however, provided a given edge is not too short, then we may be able to recover the corresponding split with high accuracy, even if the entire tree cannot be reconstructed. Finer algorithm analyses should follow from such a framework.

Our study also suggests two further theoretical questions that would be worth investigating in future. 
 Firstly, it would be of interest to analytically calculate the precise limiting stochastic safety radius of  NJ and other standard methods; in particular, to determine if it takes the value 1 or some number less than this.

It would also be of interest to study the stochastic safety radius of distance-based tree reconstruction methods for the more general class of models in which 
the $\epsilon_{xy}$ values have  a multivariate normal distribution, with means of $0$ and  a covariance matrix $\Sigma$
(the OLS model we considered in this paper assumes that $\Sigma$ is the diagonal matrix with each diagonal entry equal to $\sigma^2$).  In such models, the variance of $\epsilon_{xy}$  would typically increase 
with the path length between leaves $x$ and $y$ in the tree (the weighted least squares (WLS) assumption \cite{fit}), while the covariance for two pairs of taxa would typically  increase with the total length of the shared branches that are present in both paths connecting each pair (the generalized least-squares (GLS) assumption \cite{bul}).

\section{Acknowledgments} 
MS thanks the Allan Wilson Centre and the NZ Marsden Fund for supporting this work.

\section{Appendix: Proof of (\ref{helps})}

Substituting  $t=x+u, u \geq 0$ in  $\PP(N(0,1)>x) = \int_x^\infty \frac{1}{\sqrt{2\pi}} e^{-t^2/2} dt$ gives:
$$\PP(N(0,1)>x) = e^{-x^2/2}\int_0^\infty \frac{1}{\sqrt{2\pi}} e^{-xu}e^{-u^2/2} du < e^{-x^2/2}\int_0^\infty \frac{1}{\sqrt{2\pi}} e^{-u^2/2} du,$$
where the  second inequality is from $e^{-xu}< 1$ for all $x,u>0$. Since the last term on the right is $\frac{1}{2}$, we get the inequality in (\ref{helps}).
Turning to the asymptotic relationship, consider:
\begin{equation}
\label{limeq}
\lim_{x \rightarrow \infty} \frac{\frac{1}{\sqrt{2\pi}}\int_x^\infty e^{-t^2/2} dt}{\frac{1}{x\sqrt{2\pi}} e^{-x^2/2}}.
\end{equation}
Since the numerator and denominator limits are both zero, we can apply L'H\^opital's rule. Straightforward calculus (using the fundamental theorem of calculus for the numerator) establishes
that the limit in (\ref{limeq}) equals 1.
\hfill$\Box$


\begin{thebibliography}{}

\bibitem{att} Atteson, K.: The performance of neighbor-joining methods of phylogeny reconstruction.  Algorithmica 25(2--3): 251--278 (1999)




\bibitem{berry} Berry, V. and Gascuel, O.: Inferring evolutionary trees with strong combinatorial evidence. Theor. Comput. Sci. 240(2): 271--298 (1997)


\bibitem{bor} Bordewich, M. and Mihaescu, R.: Accuracy guarantees for phylogeny reconstruction algorithms based on balanced minimum evolution.  In Moulton, V. and Singh, M. (eds) Proceedings of WABI 2010, 10th International Workshop on Algorithms in Bioinformatics, volume 6293 of LNBI, pp 250--261. Springer-Verlag (2010)

\bibitem{bul}  Bulmer, M.: Use of the method of generalized least-squares in reconstructing phylogenies from sequence data.  Mol. Biol. Evol.  8: 868--883 (1991)

\bibitem{cas} Casella, G.  and Berger, R. L.: Statistical Inference. Duxbury Press, Belmont, CA (1990)

\bibitem{cav} Cavalli-Sforza, L.L., and Edwards, A.W.F.: Phylogenetic analysis: Models and estimation procedures. Amer. J. Hum. Genet. 19: 223--257 (1967)

\bibitem{des} Desper R., Gascuel O.: Fast and accurate phylogeny reconstruction algorithms based on the minimum-evolution principle.  J. Comput. Biol.  9: 687--706 (2002)


\bibitem{eic} Eickmeyer, K. Huggins, P., Pachter, L. and Yoshida, R.: On the optimality of the neighbor-joining algorithm.  Algorithms Molec. Biol. 3:5 (2008)

\bibitem{fit} Fitch, W.M., and Margoliash, E.:  Construction of phylogenetic trees. Science 155:279--284 (1967)


\bibitem{gas97} Gascuel O., Concerning the NJ algorithm and its unweighted version, UNJ.  {\em In} Mathematical Hierarchies and Biology, B. Mirkin, F.R. McMorris, F.S. Roberts and A. Rzhetsky (Eds.), American Mathematical Society, Providence, 149--170 (1997)

\bibitem{gas} Gascuel, O. and McKenzie, A.: Performance analysis of hierarchical clustering algorithms. J. Classif. 21: 3--18 (2004)

\bibitem{gas2} Gascuel, O. and Steel, M.: Neighbor-Joining revealed. Mol. Biol. Evol.  23(11): 1997--2000 (2006)

\bibitem{gui}  Guiasu, S.: Information Theory with Applications. McGraw-Hill, New York (1977)

\bibitem{kan} Kannan S.K., Lawler E.L., Warnow T.J.: Determining the evolutionary tree using experiments.  J. Algorithms 21:26--50 (1996)

\bibitem{mih}  Mihaescu, R.,  Levy, D. and  Pachter, L.: Why neighbor-joining works. Algorithmica, 54(1): 1--24 (2009)



\bibitem{par2} Pardi, F.,  Guillemot, S. and Gascuel, O.: Robustness of phylogenetic inference based on minimum evolution. Bull.  Math. Biol. 72, 1820--1839 (2010)

\bibitem{pea} Pearl, J. and Tarsi, M.: Structuring causal trees. J. Complexity, 2, 60--77 (1986)

\bibitem{pau} Pauplin Y.: Direct calculation of a tree length using a distance matrix, J. Molec. Evol.  51:41--47 (2000)

\bibitem{rob} Robinson D. R., Foulds L. R.: Comparison of phylogenetic trees, Math. Biosci.  53:131--147 (1981)

\bibitem{rzh}  Rzhetsky A., Nei M.: Theoretical foundation of the minimum-evolution method of phylogenetic inference. Mol. Biol. Evol.  10:1073--1095, (1993)


\bibitem{sai} Saitou, N., Nei, M.: The neighbor-joining method: A new method for reconstructing phylogenetic trees. Mol. Biol. Evol.  4:406--425 (1987)

\bibitem{sat}  Sattath, S., and Tversky, A.: Additive similarity trees, Psychometrika, 42: 319--345 (1997)





\bibitem{zar} Zarestkii K.: Reconstructing a tree from the distances between its leaves (In Russian). Uspehi Mathematicheskikh Nauk 20:90--92 (1965)

\end{thebibliography}
\end{document}